\documentclass[letterpaper, 10 pt, conference]{ieeeconf}  

\IEEEoverridecommandlockouts                              

\usepackage{amsmath,amssymb,amsfonts,mathtools}
\usepackage{graphicx}
\usepackage{epstopdf}
\usepackage{psfrag}
\usepackage{hyperref}

\DeclarePairedDelimiter{\floor}{\lfloor}{\rfloor}
\newcommand{\D}{\mathbf {D}}
\newcommand{\Op}{\mathbf {Op}}

\newcommand{\T}{\mathcal {T}}
\newcommand{\B}{\mathcal {B}}
\newcommand{\A}{\mathcal {A}}

\newtheorem{theorem}{Theorem}

\title{\LARGE \bf Maximally persistent connections for the periodic type}

\author{\authorblockN{Abdul Basit Memon\authorrefmark{1}
and Erik I. Verriest\authorrefmark{2}}
\authorblockA{School of Electrical and Computer Engineering\\
Georgia Institute of Technology, Atlanta, Georgia, USA\\
\authorrefmark{1}\tt\small abmemon@gatech.edu,\authorrefmark{2}\tt\small erik.verriest@ece.gatech.edu}
}

\begin{document}
\maketitle
\thispagestyle{empty}
\pagestyle{empty}

\begin{abstract}
This paper considers the optimal control problem of connecting two periodic trajectories with maximal persistence. A maximally persistent trajectory is close to the periodic type in the sense that the norm of the image of this trajectory under the operator defining the periodic type is minimal among all trajectories. A solution is obtained in this paper for the case when the two trajectories have the same period but it turns out to be only piecewise continuous and so an alternate norm is employed to obtain a continuous connection. The case when the two trajectories have different but rational periods is also solved. The problem of connecting periodic trajectories is of interest because of the observation that the operating points of many biological and artificial systems are limit cycles and so there is a need for a unified optimal framework of connections between different operating points. This paper is a first step towards that goal.\\
\end{abstract}

\section{INTRODUCTION}
The problem considered in this paper is that of finding a maximally persistent connection between two periodic trajectories $x^a$ and $x^b$ over a finite time interval $[a,b]$. This problem can also be stated as finding a trajectory $x(t)$ such that $x=x^a$ for $t\leq a$ and $x=x^b$ for $t\geq b$ and the trajectory in the interval $[a,b]$ is maximally persistent or as close to periodic as possible with respect to a specified norm. The problem of finding connections for periodic trajectories is one specific case in a much broader class of problems explored in earlier publications: \cite{verriest2012mtns,deryck2009persistence,verriest2008first,deryck2011thesis}, and \cite{memon2014kernel}. The problem will be rigorously stated later in the paper but before doing that we motivate why this special periodic case is of interest.

Periodic phenomena are highly prevalent in natural as well as artificial systems. For instance many biological processes ranging from the beating of the heart to locomotion, occur with periodic patterns. Moreover a single system may exhibit different types of periodic behavior. This leads to the natural question of transitions between different periodic behaviors. Consider the example of animal locomotion where the periodic patterns of movement of the limbs are called gaits. Most animals employ a variety of gaits such as one for walking and a different one for running \cite{golubitsky2003symmetry}. To switch from one gait to another, one necessarily has to employ an aperiodic transition but animals do this naturally in a graceful manner. It is our hypothesis that this translates to the transient motion remaining as close as possible to a periodic behavior. The theory of finding persistent transitions may also be of use in the control of legged robots \cite{clark2006gaits}. A popular approach to legged robot control is to specify the gaits or different schemes of motion of a robot and then switch through these gaits. This reduces the complexity of the control problem. The problem then becomes one of finding a suitable gait transition that connects the two desired gaits from the set of dynamically consistent transitions. Chemical reactors may also operate in periodic cycles since it results in better yields. In chemical process control, it is desired to transfer from one operating point (periodic cycle) to another smoothly so as to avoid drastic changes \cite{bailey1974periodic}. It can be argued here that the transition has to be maximally persistent in periodicity. These are but a few examples demonstrating the significance of the study of persistent connections between periodic trajectories.

The problem of connections for the periodic type was treated in \cite{yeung2010periodic} and \cite{yeung2011periodic}. The approach proposed in these earlier works is to solve this problem using Fourier series expansions. However to use this method practically the Fourier series needs to be truncated to a finite number of terms and so a compromise has to be made between accuracy and computational complexity, when choosing the number of terms. Another approach uses impulsive approximations to arrive at a result. The major contribution of this paper is that it follows the abstract formulation detailed in \cite{memon2014kernel} that leads to an alternative but simpler solution for finding a connection between two periodic trajectories of the same period. The connection is found over an interval of arbitrary length. It turns out that the connection in general is only piecewise continuous and so a method is also proposed to impose continuity on the connection. The case of finding connections between periodic trajectories of different periods is also discussed and a solution is presented for the case when the two periods are rational.

The rest of the paper is organized as follows: A brief review of the behavioral approach of Willems is presented in Section \ref{sec:ba}, to introduce the terminology used in the subsequent sections. Following that, in Section \ref{sec:gf}, a brief review of the abstract framework from \cite{memon2014kernel} is presented and the problem is formally stated using this framework. The main results, construction of the connections for the periodic type, are presented in Section \ref{sec:periodic} followed by some examples in Section \ref{sec:ex}. \\

\section{BEHAVIORAL APPROACH - A REVIEW}
\label{sec:ba}
We start by reviewing some of the relevant concepts from the behavioral approach to system theory. These ideas will be used in the later sections. A detailed exposition of the subject can be found in \cite{willems2007open} and \cite{willemBook}. Let $\mathbb{T}$ denote the set of independent variables called the time axis. For continuous time systems we take $\mathbb{T}=\mathbb{R}$. 
Let $\mathbb{W}$ be the set in which an $n$-dimensional observable signal vector, $w$, takes its values and is called the signal space. Typically, $\mathbb{W}=\mathbb{R}^n$, $n\geq 1$. The universum is the collection of all maps from the time axis to signal space, denoted by $\mathbb W^{\mathbb T}$. A dynamical system, $\Sigma$, is defined as a triple
$\Sigma=(\mathbb{T},\mathbb{W},\mathcal{B})$. The behavior $\mathcal{B}$ is a suitable subset of $\mathbb{W}^\mathbb{T}$, for instance the piecewise smooth functions, compatible with the laws governing $\Sigma$. We define the evaluation functional $\sigma_t$ by $\sigma_t(w)=w(t)$ a.e. (exception where $w$ is not defined). The shift operator $\mathbf S_\tau$ is defined by $\sigma_t (\mathbf S_\tau w)= \sigma_{t+\tau}w$.

The dynamical system $\Sigma=(\mathbb{T},\mathbb{W},\mathcal{B})$ is said to be linear if $\mathbb{W}$ is a vector space over $\mathbb{R}$ or $\mathbb{C}$, and the behavior $\mathcal{B}$ is a linear subspace of $\mathbb{W}^\mathbb{T}$. The dynamical system $\Sigma=(\mathbb{T},\mathbb{W},\mathcal{B})$ is said to be shift invariant   if $w\in \mathcal{B}$ implies $\mathbf S_\tau w\in \mathcal{B}$
for all $\tau\in \mathbb{T}$. If $\Sigma=(\mathbb{T},\mathbb{W},\mathcal{B})$ is a shift-invariant dynamical system,  the behavior restricted to a small open interval $(-\epsilon,\epsilon)$ is defined by
$\mathcal{B}_\epsilon=\{\tilde{w}:(-\epsilon,\epsilon)\rightarrow \mathbb{W}\,|\,\exists w\in \mathcal{B}\; {\rm such}\;{\rm that}\;\sigma_t{\tilde{w}}=\sigma_t{w}\; {\rm for}\; {\rm all}\; -\epsilon<t<\epsilon\}.$ The continuous time system $\Sigma$ is called {\em locally specified} if for all $\epsilon>0$,\[(w\in \mathcal{B})\Leftrightarrow (\left.\mathbf S_\tau w\right|_{(-\epsilon,\epsilon)}\in \mathcal{B}_\epsilon\; {\rm for}\;{\rm all} \;\tau\in \mathbb{R}).\]
The behavior defined by the system of differential equations
\[R(\D)w=0,\quad R(\xi)\in \mathbb{R}^{p\times n}[\xi]\]
where $R(\xi)$ is a matrix of polynomials with real coefficients and $\D$ is the differentiation operator, represents a system of $p$ linear time invariant (LTI) ordinary differential equations (ODE) in $n$ scalar variables. A system described by behavioral differential equations is locally specified. In order to verify if a trajectory $w$ belongs to the behavior, it suffices to look at the trajectory
in an infinitesimal neighborhood about each point.

A behavior is called {\em autonomous} if for all $w_1,w_2\in\mathcal{B}$ $w_1(t)=w_2(t)  \; {\rm for}\; t\leq 0$ implies $w_1(t)=w_2(t)$ for almost all $ t$. For an autonomous system, the future is entirely determined by its past. The notion of {\em controllability} is an important concept in the behavioral theory. Let $\mathcal{B}$ be the behavior of a linear time invariant system.  This system is called controllable if for any two trajectories $w_1$ and $w_2$ in  $\mathcal{B}$, there exists a $\tau\geq 0$ and a trajectory $w\in \mathcal{B}$ such that
\[
\sigma_t(w)=\left\{
\begin{array}{ll}\sigma_t(w_1) & t\leq 0\\\sigma_t(\mathbf S_{-\tau}w_2)& t\geq \tau\end{array}\right.\]
i.e., one can switch from one trajectory to the other, with perhaps a delay, $\tau$. Note that an autonomous system cannot get off a trajectory once it is on it. Hence an autonomous system is not controllable. \\

\section{GLUSKABI FRAMEWORK - A REVIEW}
\label{sec:gf}
As mentioned earlier, the objective of this paper is to find a connection between two periodic trajectories, say $x^a$ and $x^b$, over a finite interval in such a way that the connection is as periodic as possible. But before we delve into this specific problem we will review a generalized framework, named the Gluskabi framework, for dealing with problems of this type. This framework was first introduced in \cite{verriest2012mtns} and a detailed exposition can be found in \cite{memon2014kernel}. The more generalized problem is that of finding a trajectory $x(t)$ in the base behavior, given two trajectories $x^a$ and $x^b$ of the same type and a finite time interval $[a,b]$, such that $x=x^a$ for $t\leq a$ and $x=x^b$ for $t\geq b$ and $x$ in the interval $[a,b]$ is maximally persistent in the type. The following definitions will clarify this statement. 

The {\bf Base Behavior} ($\B_0$) is a subset of the universum $\B_0\subset \mathbb W^\mathbb T$ that defines the set of all allowable functions of interest. For any particular problem, the functions we are trying to connect lie in this set and the search for a connection\footnote{This usage of the term connection is different from a connection defined in differential geometry.} between the two is also conducted in this set. 

A \textbf{Type ($\T$)} is a strict subset of the base behavior ($\T\subset\B_0$) described by an operator $\Op:\A\to \mathcal V$ in the following way:
	\[ \mathcal T =\{w\in\mathcal A \;\text{such that}\; {\Op}\,w = 0 \}\]
where $\A\subset\B_0$ is the maximal linear space in the base behavior on which the operator is properly defined $\A\subset Dom(\Op)$ and $\mathcal V$ is a linear space as well. This is a familiar kernel representation.

A \textbf{Trait ($\T_\theta$)} is a subtype of the type i.e., it is a subset of the type such that it has its own characteristic behavior, given by some operator $\Op_\theta$, parametrized by $\theta$:
\[ \mathcal T_\theta =\{w\in\mathcal T \;\text{such that}\; {\Op}_{\theta}w = 0 \} \] 

Given a type $\T$ with the associated operator $\Op$, an element $w\in\A\subset\B_0$ is said to be \textbf{maximally persistent} with respect to the norm $\|.\|$, defined on $\mathcal V$ restricted to $[a,b]$, if $w$ minimizes $\|\Op\,w\|$.

Given a type $\T$ with the associated operator $\Op$, the \textbf{Gluskabi map} $g:\T\times\T\to\B_0$ with respect to the norm $\|.\|$, defined on $\mathcal V$ restricted to $[a,b]$, is defined as follows 
\[g(w_1,w_2)(t)=\left\{\begin{array}{lr}w_1(t) & t\leq a \\ arg\!\min_{w\in\A}{\|\Op\,w\|} \quad & a<t<b\\ w_2(t) & t\geq b\end{array}\right.\]

Clearly this requires that $\mathcal V$ restricted to the interval $[a,b]$ be a normed space. The connection in the interval $[a,b]$ will be called the ``\emph{Gluskabi raccordation}". Now that prerequisite terminology has been defined, it can be formally stated that the objective of this paper is to find the Gluskabi raccordation between any two given trajectories of the periodic type or in other words to find the Gluskabi map for the periodic type. \\

\section{THE PERIODIC TYPE}
\label{sec:periodic}
In the context of the framework from the previous section, the periodic types are defined by the kernel of operators involving the shift operator. For instance the $\tau$-periodic type, which is the behavior of periodic functions of period $\tau$, is defined by the kernel of the operator $\Op:= (\mathbf I - \mathbf S_\tau)$ where $\mathbf I$ is the identity operator and $\mathbf S$ is the shift operator. The shift operator can either be defined as an advance - $\mathbf S_\tau f(t) = f(t+\tau)$ or as a lag - $\mathbf S_\tau f(t) = f(t-\tau)$. Throughout this paper we will choose the shift to be a lag. A periodic function whose Fourier series exists can also be seen as a sum of harmonic signals of integer multiples of the base frequency. Inspired by this observation, the $\tau$-periodic type in the base behavior $\B_0=C^\omega(\mathbb R,\mathbb R)$ may also be characterized by the infinite product operator $\left[\D\prod_{n=1}^\infty{\left(1+\frac{1}{n^2\omega^2}\D^2\right)}\right]$, which can also be written as $\sinh{\left(\frac{\pi}{\omega}\D\right)}$ \cite{silverman1984ICA}, where $\omega=2\pi/\tau$. This representation defines a number of traits in terms of the number of finite product terms and these traits serve as various levels of approximation to the periodic functions. The equivalence between the two operators can be shown with some work 
but the former more general operator will be employed for the following theorems.

Three results are presented in this section for finding the Gluskabi raccordation between two trajectories of the periodic type. The first result deals with finding the raccordation for the $\tau$-periodic type over an interval which is a multiple of the period. The second result generalizes this and the raccordation interval can be of arbitrary length. The third result yields continuous raccordations. Finally, the case of raccordations between trajectories of different periods is discussed.
\subsection{Gluskabi Raccordation}
\begin{theorem}
Given two trajectories $x^a$ and $x^b$ from the $\tau$-periodic type, with the associated operator $(\mathbf I - \mathbf S_\tau)$, the Gluskabi raccordation between the two over the interval $[a,b]$ with respect to the $L^2$ norm is given by 
\begin{IEEEeqnarray*}{rCl}
	x(t) &=& x^a(t) + \left(1+\floor*{\frac{t-a}{\tau}}\right)\frac{x^b(t)-x^a(t)}{n+1}		
	\end{IEEEeqnarray*}
provided $b = a + n\tau$ for some $n\in\mathbb Z_+$.
\label{thm:multiple_period}
\end{theorem}
\begin{proof}
This can be proved using optimal control theory. Let $u = (\mathbf I - \mathbf S_\tau)x$ or \[u(t) = x(t)-x(t-\tau)\] for any $x$ in the domain of the operator. Notice that if $x$ is periodic then $u=0$. Then the Gluskabi raccordation is the argument of the following optimization problem:
\begin{equation}
	\min_{x(t)}{J} = \min_{x(t)}{\frac{1}{2}\int_a^{b+\tau}{u^2(t)dt}}. 
	\label{eq:opt_prob}
\end{equation}
Notice that the integral is taken over the interval $[a,b+\tau]$ instead of $[a,b]$. This is the maximal interval over which $u$ is possibly non-zero since $x$ is equal to the $\tau$-periodic functions $x^a$ and $x^b$ over the intervals $(-\infty,a]$ and $[b,\infty)$ respectively and so $u$ is zero over the intervals $(-\infty,a]$ and $[b+\tau,\infty)$. Now adjoining the definition of $u$ with Lagrange multiplier to the cost function (\ref{eq:opt_prob}) gives:
\[ J(u)=\int_a^{b+\tau}{ \frac{u^2(t)}{2}+\lambda(t)\left[x(t)-x(t-\tau)-u(t)\right]dt}. \]
Perturbing $u$ by $\delta u$ changes the cost function to 
\begin{multline}
	J(u+\delta u)=\int_a^{b+\tau}\frac{1}{2}(u+\delta u)^2+\lambda(t)\left[x(t)+\delta x(t)\right. \\ \left.-x(t-\tau) -\delta x(t-\tau)-u(t)-\delta u(t)\right]dt
\end{multline}
and so 
\begin{IEEEeqnarray}{rCl}
\delta J &= & J(u+\delta u)-J(u)\nonumber \\
&\approx& \int_a^{b+\tau}u(t)\delta u(t)+\lambda(t)\left[\delta x(t)-\delta x(t-\tau)\right.\nonumber \\ 
&& \left.\quad\quad -\delta u(t)\right]dt \nonumber\\
&=& \int_a^{b+\tau}{[u(t)-\lambda(t)]\delta u(t)dt} + \int_a^{b+\tau}{\lambda(t)\delta x(t)dt} \nonumber\\
&& -\int_{a-\tau}^{b}{\lambda(t+\tau)\delta x(t)dt}. \label{eq:J_gateaux}
\end{IEEEeqnarray}
A necessary condition for an $x$ to minimize (\ref{eq:opt_prob}) is that $\delta J$ be zero for any arbitrary $\delta u$. The perturbation $\delta x$ in (\ref{eq:J_gateaux}) is zero in the intervals $[b,b+\tau]$ and $[a-\tau,a]$ and so the integrals involving $\delta x$ are zero in these intervals. The multiplier $\lambda(t)$ is chosen in the following way in the interval $(a,b)$ to avoid computing $\delta x$ in this interval: 
\begin{equation}
	\lambda(t)=\lambda(t+\tau) \quad\quad \forall t \in (a,b).
	\label{eq:euler}
\end{equation}
Then the necessary condition for optimality is
\begin{equation}
	\lambda(t)=u(t) \quad\quad \forall t \in [a,b+\tau].
	\label{eq:optimality}
\end{equation}
Notice that $\lambda(t)$ is free in the interval $[b,b+\tau]$ and as a consequence of (\ref{eq:euler}) choosing $\lambda$ in any $\tau$ length sub-interval in $[a,b+\tau]$ completely determines it for all time. The initial definition of $u$ in the interval $[a,b+\tau]$ can be written as follows, using (\ref{eq:optimality}) and (\ref{eq:euler}): For any $\theta\in[a,a+\tau]$ and $i\in\{0,1,\cdots,n\}$,
\begin{IEEEeqnarray}{rCl}
	x(\theta+i \tau) &=& u(\theta+i\tau) + x(\theta + (i-1)\tau) \nonumber\\
	&=& u(\theta+i\tau)+u(\theta + (i-1)\tau) + x(\theta + (i-2)\tau) \nonumber\\
	&=& (i+1)\;u(\theta) + x(\theta -\tau) \nonumber\\
	&=& (i+1)\;u(\theta) + x^a(\theta-\tau).
	\label{eq:difference}
\end{IEEEeqnarray}
When $i=n$, we have that $x(\theta+n \tau)=x^b(\theta+n \tau)$ and so (\ref{eq:difference}) becomes
\begin{equation*}
	x^b(\theta+n \tau) = (n+1)\;u(\theta) + x^a(\theta-\tau).
\end{equation*}
Since $x^a$ and $x^b$ are both periodic with period $\tau$, the above equation yields the following simplified expression for $u(\theta)$ in the interval $[a,a+\tau]$:
\begin{equation}
u(\theta) = \frac{x^b(\theta)-x^a(\theta)}{n+1}.		
\end{equation}
This consequently defines $u(t)$ in the entire interval $[a,b+\tau]$ by (\ref{eq:euler}) and (\ref{eq:optimality}). Thus the Gluskabi raccordation for this periodic type is obtained from (\ref{eq:difference}) and is as follows: 
\[x(t) = x^a(t) + \left(1+\floor*{\frac{t-a}{\tau}}\right)u(t).\] 
\end{proof}

This result indicates that the Gluskabi raccordation basically takes the difference between a period of the two trajectories and covers this difference during the raccordation interval in $n+1$ periods. This result is generalized in the next theorem for the case where the raccordation interval is of arbitrary length i.e., its length is not restricted to be a multiple of the period. \\

\begin{theorem}
Given two trajectories $x^a$ and $x^b$ from the $\tau$-periodic type, with the associated operator $(\mathbf I - \mathbf S_\tau)$, the Gluskabi raccordation between the two over the interval $[a,b]$ with respect to the $L^2$ norm is given by 
\begin{IEEEeqnarray*}{rCl}
	x(t) &=& x^a(t) + \left(1+\floor*{\frac{t-a}{\tau}}\right)u(t)
\end{IEEEeqnarray*}
where 
\begin{IEEEeqnarray*}{rCl}
u(t) = \left\{ \begin{array}{lr} \frac{1}{n+2}\left[x^b(t)-x^a(t)\right] & a\leq t' \leq (b-n\tau)	\\
										\frac{1}{n+1}\left[x^b(t)-x^a(t)\right] & (b-n\tau)\leq t'\leq (a+\tau) \\ & t' = (t-a) \mod \tau\end{array}\right. 
\end{IEEEeqnarray*}
\label{thm:general}
\end{theorem}
\begin{proof}
The proof of this theorem is exactly along the lines of Theorem \ref{thm:multiple_period} since we employ the same cost function and we obtain the same Euler-Lagrange equation and the optimality equation i.e.,
\begin{IEEEeqnarray}{rClr}
	\lambda(t)&=&\lambda(t+\tau)  \quad&\forall t \in (a,b) \label{eq:euler_thm2}\\
	\lambda(t)&=& u(t)  \quad&\forall t \in [a,b+\tau]. \label{eq:optimality_thm2}
\end{IEEEeqnarray}
Again it holds that $\lambda(t)$ is free in the interval $[b,b+\tau]$ and as a consequence of (\ref{eq:euler_thm2}) choosing $\lambda$ in any $\tau$ length sub-interval in $[a,b+\tau]$ completely determines it for all time. The raccordation interval $[a,b]$ can be split up into intervals of length $\tau$ with possibly one remaining interval of length less than $\tau$. Let $n=\floor*{\frac{b-a}{\tau}}$. Then $b-a = n\tau + (b-a-n\tau)$. From the definition of $u$ we have that,
\[ x(t) = u(t) + x(t-\tau). \]
Or for any $\theta\in[a,a+\tau]$ and $i\in\{0,1,\cdots,n\}$,
\begin{IEEEeqnarray}{rCl}
	x(\theta+i \tau) &=& u(\theta+i\tau) + x(\theta + (i-1)\tau) \nonumber\\
	&=& (i+1)\;u(\theta) + x^a(\theta-\tau)
	\label{eq:difference_thm2}
\end{IEEEeqnarray}
by making use of (\ref{eq:euler_thm2}) and (\ref{eq:optimality_thm2}). There are two separate cases to be dealt with here. When $i=n$ and $\theta\in [b-n\tau,a+\tau]$, we have that $x = x^b$ and so (\ref{eq:difference_thm2}) becomes
\begin{equation}
	x^b(\theta+n\tau) = (n+1)u(\theta) + x^a(\theta-\tau).
	\label{eq:diff_right}
\end{equation}
On the other hand when $i=n+1$ and $\theta\in [a,b-n\tau]$, (\ref{eq:euler_thm2}) still holds since the argument $(\theta+i\tau)$ is in the interval $[b,b+\tau]$ and so (\ref{eq:difference_thm2}) becomes
\begin{equation}
	x^b(\theta+(n+1)\tau) = (n+2)u(\theta) + x^a(\theta-\tau).
	\label{eq:diff_left}
\end{equation}
Using the fact that both $x^a$ and $x^b$ are periodic with period $\tau$, the above equations (\ref{eq:diff_right}) and (\ref{eq:diff_left}) yield the following simplified expression for $u(\theta)$ in the interval $[a,a+\tau]$:
\begin{equation}
u(\theta) = \left\{ \begin{array}{lr} \frac{1}{n+2}\left[x^b(\theta)-x^a(\theta)\right] & a\leq\theta\leq (b-n\tau)	\\
										\frac{1}{n+1}\left[x^b(\theta)-x^a(\theta)\right] & (b-n\tau)\leq\theta\leq (a+\tau)\end{array}\right..
\label{eq:diff_final}
\end{equation}
This consequently defines $u(t)$ in the entire interval $[a,b+\tau]$ by (\ref{eq:euler_thm2}) and (\ref{eq:optimality_thm2}). Thus the Gluskabi raccordation over the interval $[a,b]$ is obtained by substituting the above expression in (\ref{eq:difference_thm2}) as follows:
\begin{IEEEeqnarray*}{rCl}
x(t) &=& x^a(t) + \left(1+\floor*{\frac{t-a}{\tau}}\right)u(t).	
\end{IEEEeqnarray*}
\end{proof}

As in Theorem \ref{thm:multiple_period}, this result indicates that the Gluskabi raccordation is the difference between a period of the two trajectories but this time covered partly in $n+1$ steps and partly in $n+2$ steps. Before proceeding to the next result some remarks will be made about the previous two theorems. 
The same results are obtained for the Gluskabi raccordation if the alternate definition of the shift operator, $\mathbf S_\tau f(t) = f(t+\tau)$, is used in the operator describing the periodic type in theorems \ref{thm:multiple_period} and \ref{thm:general}. This indicates that the Gluskabi raccordation obtained is truly associated with the periodic type. The two theorems can also be viewed in the Fourier domain. The Gluskabi raccordation there is similar to the time domain and the difference between the Fourier coefficients is equally covered in $n+1$ or $n+2$ steps. Furthermore, these same results can also be obtained using the compact adjoint expression derived for the Gluskabi raccordation in \cite{memon2014kernel}. The present approach is chosen simply because of the insight it offers into the periodic type. 

\subsection{Continuous Gluskabi Raccordation} 
The chosen Base behavior $\B_0$ for both the previous results is the space of piecewise continuous functions and as the results indicate the Gluskabi raccordation is piecewise continuous even if the trajectories being connected are continuous. The resulting raccordation can be approximated by continuous solutions but there is no unique way of doing this. A way to impose continuity of the Gluskabi raccordation is to append a cost on the derivative of the trajectory to the usual cost function i.e., instead of minimizing $\|\Op x\|$ one can minimize $\|\Op x\|+\rho^2\|\D\,\Op x\|$ with $\rho^2$ being a weighting factor. In order to do this the base behavior $\B_0$ is chosen to be the space of continuous functions. The Gluskabi raccordation is now found for the piecewise differentiable $\tau$-periodic type and the operator characterizing it is $\Op\,x=[x(t)-x(t-\tau)]+\rho^2[\dot{x}(t)-\dot{x}(t-\tau)]$. This makes sense because if $x$ is $\tau$-periodic or $x(t)-x(t-\tau)=0$ then $\dot{x}(t)-\dot{x}(t-\tau)=0$. The resulting Gluskabi raccordation is presented in the following theorem. \\

\begin{theorem}
Given two trajectories $x^a$ and $x^b$ from the piecewise differentiable $\tau$-periodic type, with the associated operator $(\mathbf I - \mathbf S_\tau)+\rho^2(\D - \D\,\mathbf S_\tau)$, the Gluskabi raccordation between the two over the interval $[a,b]$ with respect to the $L^2$ norm is given by solving the following set of equations:
\begin{IEEEeqnarray*}{rCl}
	f_0 &=& \rho\left(x_1+x_a-2x_0\right) \nonumber\\
	f_k &=& \rho\left(x_{k+1}+x_{k-1}-2x_k\right) \quad\quad 1\leq k\leq n-2\nonumber\\
	f_{n-1} &=& \rho\left(x_b+x_{n-2}-2x_{n-1}\right) \\
	f_i &=& c_{0}^1\,e^{\frac{1}{\rho}(t+i\tau)} - c_{0}^2\,e^{-\frac{1}{\rho}(t+i\tau)} \quad 0\leq i\leq n-1
\end{IEEEeqnarray*}
with associated boundary conditions $x_0(0) = x^a(a)$ and $x_{n-1}(\tau) = x^b(b)$, provided $b = a + n\tau$ for some $n\in\mathbb Z_+$.
\label{thm:continuous}
\end{theorem}

\begin{proof}
The proof uses the usual machinery of optimal control. The Gluskabi raccordation in this case is the argument of the following optimization problem: 
\[ \min_{x(t)}{\int_a^{b+\tau}{[x(t)-x(t-\tau)]^2+\rho^2[\dot{x}(t)-\dot{x}(t-\tau)]^2dt}}. \]
To simplify the problem let's define a set of shifted functions $x_k(\theta) = x(a+k\tau+\theta)$ for $k\in\{0,\cdots,n-1\}$ and $\theta\in[0,\tau]$. Also let $u_k=\dot{x}_k$, $x_a(\theta)=x^a(a-\tau+\theta)$ and $x_b(\theta)=x^b(b+\theta)$. These shifted functions cover the entire interval $[a,b]$ and the cost function $J$ can now be written as,
\begin{multline}
J = \frac{1}{2}\int_0^{\tau} \left[x_0(\theta)-x_a(\theta)\right]^2 + \rho^2\left[u_0(\theta)-u_a(\theta)\right]^2+ \\ \quad\sum_{k=1}^{n-1}{\left[x_k(\theta)-x_{k-1}(\theta)\right]^2+\rho^2\left[u_k(\theta)-u_{k-1}(\theta)\right]^2} + \\ \left[x_{b}(\theta)-x_{n-1}(\theta)\right]^2 + \rho^2\left[u_b(\theta)-u_{n-1}(\theta)\right]^2 d\theta
\label{eq:cost_cont}
\end{multline}
and the boundary conditions take the form of $x_0(0) = x_a(\tau)$ and $x_{n-1}(\tau) = x_b(0)$. Since the raccordation $x$ has to be continuous, additional constraints are imposed on the boundaries of the interior shifted functions, specifically $x_k(0)=x_{k-1}(\tau)$ for $k\in\{1,\cdots,n-1\}$. Now adjoining the constraint equations $u_k-\dot{x_k}=0$ along with the Lagrange multiplier $\lambda_k$ for $k\in\{0,\cdots,n-1\}$ to the cost function (\ref{eq:cost_cont}) and employing the usual techniques of optimal control the following set of Euler-Lagrange equations are obtained:
\begin{IEEEeqnarray}{rCl}
\dot{\lambda}_0 &=& x_1+x_a-2x_0 \nonumber\\
\dot{\lambda}_k &=& x_{k+1}+x_{k-1}-2x_k \quad\quad 1\leq k\leq n-2\nonumber\\
\dot{\lambda}_{n-1} &=& x_b+x_{n-2}-2x_{n-1}.
\label{eq:euler_cont}
\end{IEEEeqnarray}
Boundary conditions on the lagrange multipliers are also obtained stemming from the fact that to preserve continuity the perturbations at the end points of the shifted functions are the same i.e. $\delta x_k(0) = \delta x_{k-1}(\tau)$. These boundary conditions turn out to be $\lambda_k(0) = \lambda_{k-1}(\tau)$ for $k\in\{1,\cdots,n-1\}$, which has a similar form to the boundary conditions on shifted functions. The set of optimality conditions is as follows:
\begin{IEEEeqnarray}{rCl}
\frac{\lambda_0}{\rho^2} &=& u_1+\dot{x}^a-2u_0 \nonumber\\
\frac{\lambda_k}{\rho^2} &=& u_{k+1}+u_{k-1}-2u_k \quad\quad 1\leq k\leq n-2\nonumber\\
\frac{\lambda_{n-1}}{\rho^2} &=& \dot{x}^b+u_{n-2}-2u_{n-1}. 
\label{eq:optimality_cont}
\end{IEEEeqnarray}
Differentiating (\ref{eq:euler_cont}) once and comparing it to (\ref{eq:optimality_cont}) yields the following set of second order differential equations for the Lagrange multipliers: 
\begin{equation}
	\ddot{\lambda}_k-\frac{1}{\rho^2}\lambda_k = 0 \quad\quad 0\leq k\leq n-1.
\end{equation}
The solutions $\lambda_k$ are a linear combination of exponential modes $e^{\pm t/\rho}$ and can be written as,
\begin{equation}
 \lambda_k = c_{k}^1\,e^{t/\rho} + c_{k}^2\,e^{-t/\rho} 
	\label{eq:lagrange_cont}
\end{equation}
for some constants $c_{k}^1$ and $c_{k}^2$. This results in $2n$ unknowns in the form of these constants which can be solved for using the $n-1$ boundary conditions for Lagrange multipliers, the $n-1$ boundary conditions for functions $x_k$ and the two end point conditions $x_0(0) = x_a(\tau)$ and $x_{n-1}(\tau) = x_b(0)$. From the conditions $\lambda_k(0) = \lambda_{k-1}(\tau)$ for $k\in\{1,\cdots,n-1\}$, the following condition is obtained,
\[ c_k^1 + c_k^2 = c_{k-1}^1\,e^{\tau/\rho} + c_{k-1}^2\,e^{-\tau/\rho} \]
\begin{equation}
	\Rightarrow c_k^1 = c_{k-1}^1\,e^{\tau/\rho} + c_{k-1}^2\,e^{-\tau/\rho} - c_k^2.
	\label{eq:const1}
\end{equation}
Additionally, the conditions $x_k(0)=x_{k-1}(\tau)$ for $k\in\{1,\cdots,n-1\}$ on the shifted functions, translate to boundary conditions on $\dot{\lambda_k}$ in (\ref{eq:euler_cont}) i.e., $\dot{\lambda_k}(0) = \dot{\lambda_{k-1}}(\tau)$ for $k\in\{1,\cdots,n-1\}$ or that 
\[ \frac{1}{\rho}\left(c_k^1 - c_k^2\right) = \frac{1}{\rho}\left(c_{k-1}^1\,e^{\tau/\rho} - c_{k-1}^2\,e^{-\tau/\rho}\right). \]
Substituting (\ref{eq:const1}) in this yields,
\begin{equation}
	c_k^2 = c_{k-1}^2\,e^{-\tau/\rho}
	\label{eq:const2}
\end{equation}
and then (\ref{eq:const1}) becomes
\begin{equation}
	c_k^1 = c_{k-1}^1\,e^{\tau/\rho}.
	\label{eq:const1_final}
\end{equation}
Therefore, the expressions for all the Lagrange multipliers (\ref{eq:lagrange_cont}) can be written in terms of just two constants $c_0^1$ and $c_0^2$ as follows:
\begin{equation}
	\lambda_k = c_{0}^1\,e^{\frac{1}{\rho}(t+k\tau)} + c_{0}^2\,e^{-\frac{1}{\rho}(t+k\tau)}.
	\label{eq:lagrange2_cont}
\end{equation}
Substituting this in (\ref{eq:euler_cont}) gives us the set of equations to be solved for finding the Gluskabi raccordation. 
\end{proof}

\subsection{Raccordation for differing periods}
The problem of finding the connection between two periodic trajectories of different periods in general is more complicated. The complication arises from the fact that there is no simple operator such that the kernel of this operator is the behavior of all periodic trajectories of arbitrary period. The type of periodic trajectories with period belonging to a compact interval can be characterized by an operator involving the minimization operation but that does not lead to a direct easy solution. An alternate approach for connecting two periodic trajectories with rational periods is to choose the operator that characterizes the periodic type with period equal to the least common multiple of the periods of the two trajectories. In other words, given two trajectories $x^a$ and $x^b$ from the $\tau_1$- periodic type and $\tau_2$-periodic type the operator $\Op:=\mathbf I - \mathbf S_\tau$ can be used where $\tau = lcm(\tau_1,\tau_2)$, provided that both $\tau_1$ and $\tau_2$ are rationals. Notice that this $\tau$-periodic type includes both the $\tau_1$-periodic type and $\tau_2$-periodic type. Once the operator has been found any of the previous results can be applied to obtain the relevant Gluskabi raccordation when the two periods are different but rational. \\

\section{EXAMPLES}
\label{sec:ex}
In this section, we illustrate the results presented in the previous section with the help of some examples. The first example is the problem of finding raccordations between $\cos{2\pi t}$ and the triangle wave with period one. Notice that both these trajectories are from the periodic type with period one and this type can be characterized by the operator $(\mathbf I-\mathbf S_\tau)$ where $\tau=1$. The raccordation is sought over the interval $[0,2.5]$ and so the length of the interval is not a multiple of the period. The result from Theorem \ref{thm:general} is applied and the resultant raccordation is shown in Fig. \ref{fig:nonmultiple}. Notice the discontinuities at multiples of the period as well as at distance $2.5\mod{1}=\frac{1}{2}$ or at the midpoint within each period. The discontinuities at the multiples of the period is due to the discrepancy in the values of the functions at the end points of a single period i.e. the cosine function at $0$ is one but the triangle wave being considered is zero at $0$. The discontinuity at the middle point in every period can be attributed to the discrepancy in the values of the two functions at that particular point in every period. Therefore considering one aligned period, if the two function are equal at the end points and at a distance equal to the raccordation interval length modulo the period then the raccordation will be continuous. 

\begin{figure}[h]
\centering
\includegraphics[scale=0.4]{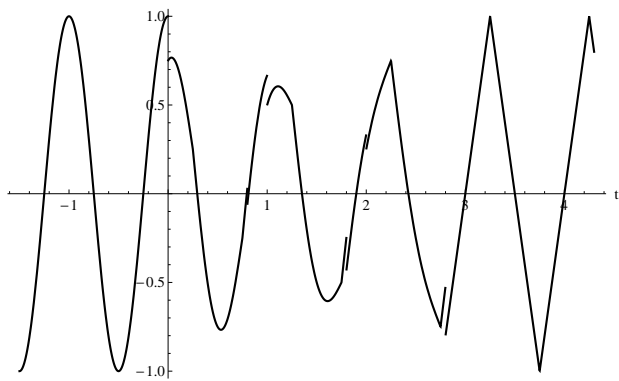}
\caption{The function is $\cos{2\pi t}$ for $t\leq 0$ and a triangle wave for $t\geq 2.5$. Raccordation is over the interval $[0,2.5]$.}
\label{fig:nonmultiple}
\end{figure}

The second example considers the same problem of finding raccordation between $\cos{2\pi t}$ and the triangle wave but now a continuous raccordation is sought over the interval $[0,4]$. The result from Theorem \ref{thm:continuous} will be used in this case and the resultant raccordation is depicted in Fig. \ref{fig:cont}, where the regularization factor $\rho = 1$. This regularization factor means that both the discrepancy in the trajectory and the discrepancy in the derivative are equally weighted. Increasing the value of $\rho$ would weigh the derivative more and so smoothen the raccordation. Notice the decrease in the magnitude of the raccordation followed by an increase to the right magnitude. This pinching effect can be attributed to the difference in phase of the two trajectories being connected. It was also observed previously in the problem of connecting harmonics in \cite{deryck2011thesis}. The effect is most pronounced when the two trajectories are $180$ degrees out of phase. Equivalently, there is no pinching at all when the two trajectories are phase aligned. This is illustrated in Fig. \ref{fig:cont2} and Fig. \ref{fig:cont3}.\\

\begin{figure}[h]
\centering
\includegraphics[scale=0.4]{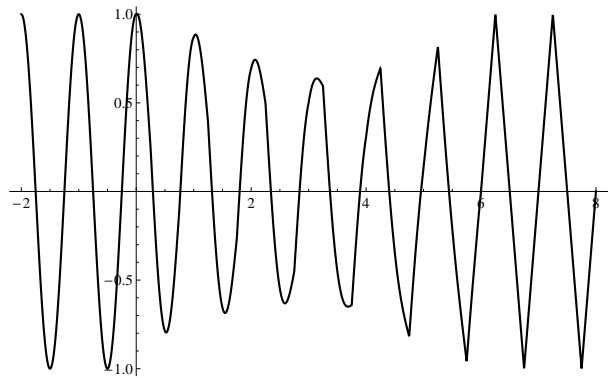}
\caption{The function is $\cos{2\pi t}$ for $t\leq 0$ and the triangle wave for $t\geq 4$. Raccordation is over the interval $[0,4]$.}
\label{fig:cont}
\end{figure}

\begin{figure}[h]
\centering
\includegraphics[scale=0.4]{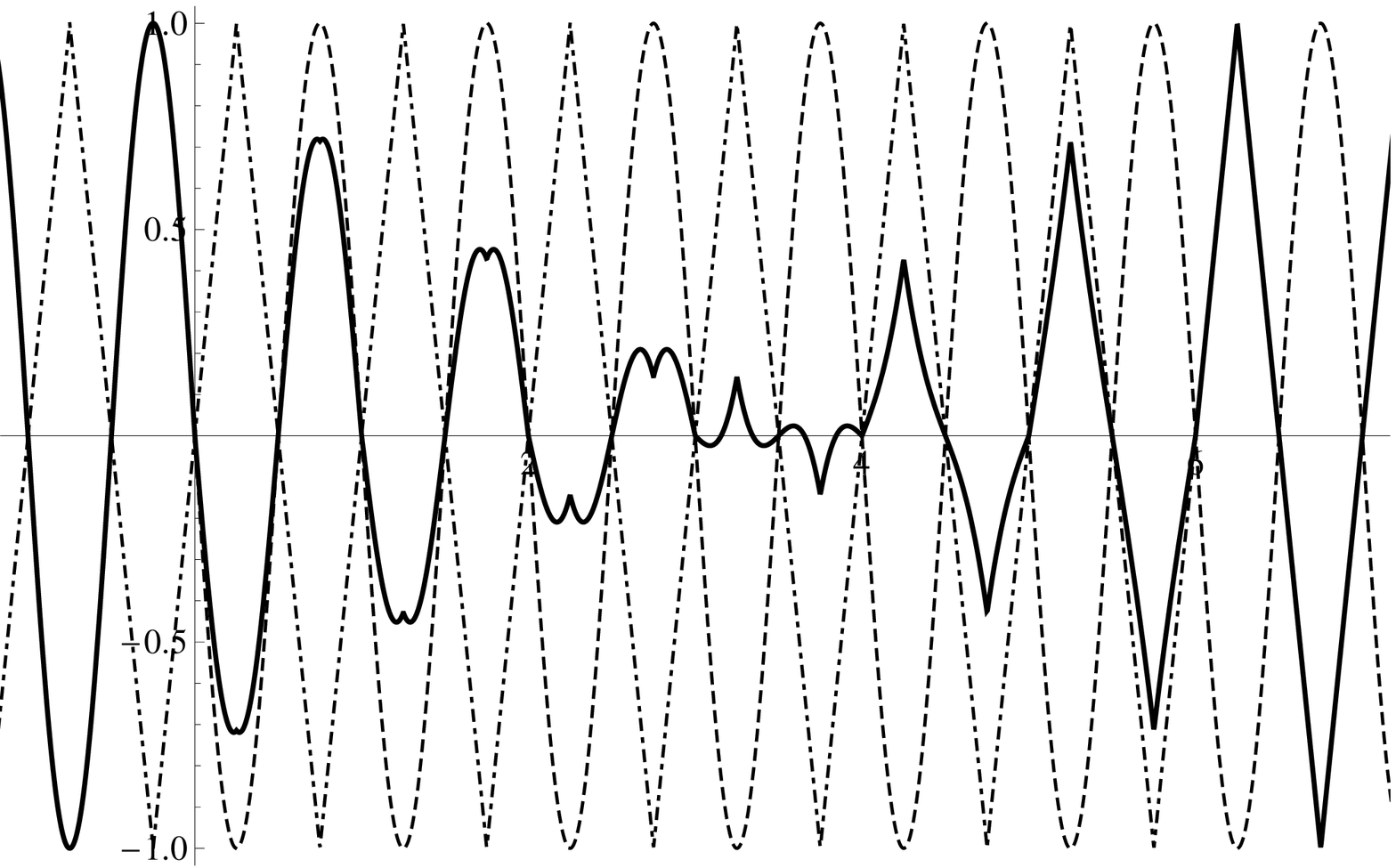}
\caption{The function is $\cos{\left(2\pi t+\frac{\pi}{2}\right)}$ for $t\leq 0$ and the triangle wave for $t\geq 6$. The two functions are $180^\circ$ out of phase. Dashed lines show the functions being connected.}
\label{fig:cont2}
\end{figure}

\begin{figure}[h]
\centering
\includegraphics[scale=0.4]{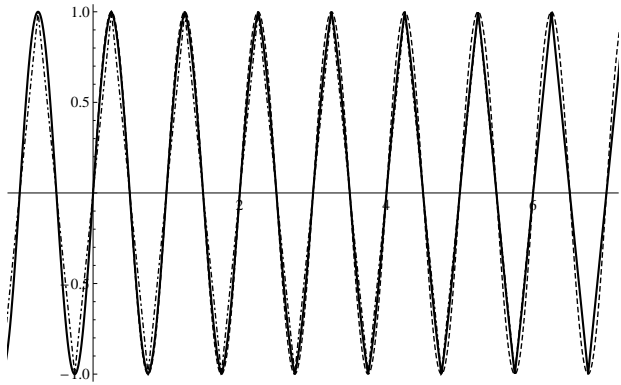}
\caption{The function is $\cos{\left(2\pi t-\frac{\pi}{2}\right)}$ for $t\leq 0$ and the triangle wave for $t\geq 6$. The two functions are phase aligned. Dashed lines show the functions being connected.}
\label{fig:cont3}
\end{figure}




%

%


\bibliographystyle{IEEEtran}
\bibliography{gluskabi}

\end{document}